\documentclass{article}
\usepackage[utf8]{inputenc}
\usepackage[letterpaper,margin=1in]{geometry}
\usepackage{tcs-macros}
\usepackage[style=alphabetic, maxalphanames=8, maxnames=8]{biblatex}
\newreptheorem{theorem}{Theorem}
\newreptheorem{lemma}{Lemma}

\newcommand{\lvl}[2][0]{#2^{(#1)}}
\newcommand{\bDel}{\mathbf{\Delta}}
\newcommand{\vst}{v_\mathsf{st}}
\newcommand{\ved}{v_{\mathsf{ed}}}
\newcommand{\BS}{\mathsf{BS}}
\newcommand{\LCA}{\mathsf{LCA}}
\newcommand{\comment}[1]{%
  \text{\phantom{(#1)}} \tag{#1}
}

\newif\ifshowauthors
\showauthorstrue

\title{Recursive Error Reduction for Regular Branching Programs}
 \ifshowauthors
 \author{  Eshan Chattopadhyay\thanks{Supported by a Sloan Research Fellowship and NSF CAREER Award 2045576.}\\ Cornell University\\ \texttt{eshan@cs.cornell.edu} \and   \and Jyun-Jie Liao\footnotemark[1] \\ Cornell University\\ \texttt{jjliao@cs.cornell.edu} }
 \date{November 26, 2023}
 \else
 \author{}
 \date{}
 \fi

\begin{document}

\maketitle
\begin{abstract}
In a recent work, Chen, Hoza, Lyu, Tal and Wu (FOCS 2023) showed an improved error reduction framework for the derandomization of regular read-once branching programs (ROBPs). Their result is based on a clever modification to the inverse Laplacian perspective of space-bounded derandomization, which was originally introduced by Ahmadinejad, Kelner, Murtagh, Peebles, Sidford and Vadhan (FOCS 2020).

In this work, we give an alternative error reduction framework for regular ROBPs. Our new framework is based on a binary recursive formula from the work of Chattopadhyay and Liao (CCC 2020), that they used to construct weighted pseudorandom generators (WPRGs) for general ROBPs. 

Based on our new error reduction framework, we give alternative proofs to the following results for regular ROBPs of length $n$ and width $w$, both of which were proved in the work of Chen et al. using their error reduction: 
\begin{itemize}
\item
There is a WPRG with error $\varepsilon$ that has seed length $$\tilde{O}(\log(n)(\sqrt{\log(1/\varepsilon)}+\log(w))+\log(1/\varepsilon)).$$ 
\item
There is a (non-black-box) deterministic algorithm which estimates the expectation of any such program within error $\pm\varepsilon$ with space complexity
$$\tilde{O}(\log(nw)\cdot\log\log(1/\varepsilon)).$$
This was first proved in the work of Ahmadinejad et al., but the proof by Chen et al. is simpler.
\end{itemize}
Because of the binary recursive nature of our new framework, both of our proofs are based on a straightforward induction that is arguably simpler than the Laplacian-based proof in the work of Chen et al. 

In fact, because of its simplicity, our proof of the second result directly gives a slightly stronger claim: our algorithm computes a $\varepsilon$-singular value approximation (a notion of approximation introduced in a recent work by Ahmadinejad, Peebles, Pyne, Sidford and Vadhan (FOCS 2023)) of the random walk matrix of the given ROBP in space $\tilde{O}(\log(nw)\cdot\log\log(1/\varepsilon))$. It is not clear how to get this stronger result from the previous proofs. 

\end{abstract}

\section{Introduction}
A central problem in complexity theory is to understand to what extent is randomness useful in space-bounded computation. It is widely conjectured that every randomized algorithm can be made deterministic with only a constant-factor blowup in space, i.e. $\BPL=\L$. A central approach to derandomize $\BPL$ is to construct explicit pseudorandom generators (PRGs) for standard-order read-once branching programs (ROBPs), which we formally define below.
\begin{definition}[ROBPs]
A (standard-order) ROBP $B$ of length $n$ and width $w$ is specified by a start state $v_0\in [w]$, a set of accept states $V_{\mathrm{acc}}$ and $n$ transition functions $B_i:[w]\times\bits{}\to {[w]}$ for $i$ from $1$ to $n$. The ROBP $B$ computes a function $B:\bits{n}\to\bits{}$ as follows. Given an input $x\in\bits{n}$, define $v_i=B_i(v_{i-1},x_i)$, where $x_i$ denotes the $i$-th bit of $x$. Then output $B(x)=1$ if $v_n\in V_{\mathrm{acc}}$, or $B(x)=0$ otherwise. 
\end{definition}
\begin{remark}
Equivalently, one can view a ROBP $B$ as a directed graph as follows. Consider $n+1$ layers of nodes $L_0,L_1,\ldots,L_n$, each having size $w$, and label the nodes in each $L_i$ with $[w]$. For every $i\in[n],v\in[w],b\in\bits{}$, construct an edge with label $b$ from $v$ in $L_{i-1}$ to $B_i(v,b)$ in $L_i$. Then the computation of $B(x)$ corresponds to a walk following label $x$ from $L_0$ to $L_n$. In this paper we usually consider the equivalent graph view, and we refer to $L_i$ as \emph{layer $i$}.   
\end{remark}

\begin{definition}[PRGs]
Let $\cF$ be a class of functions $f:\bits{n}\to\bits{}$. An $\eps$-PRG for $\cF$ is a function $G:\bits{d}\to\bits{n}$ such that for every $f\in\cF$,
$$\abs{\ex[x\sim \bits{n}]{f(x)}- \ex[s\sim\bits{d}]{f(G(s))}}\le \eps.$$ 
We say $G$ $\eps$-fools the class $\cF$ if $G$ is an $\eps$-PRG for $\cF$. We call $d$ the \emph{seed length} of $G$. We say $G$ is explicit if it can be computed in space $O(d)$.\footnote{Throughout this paper, when we say a function $f$ is explicit, it means the function $f$ can be computed in space $O(n)$ where $n$ is the input length.} 
\end{definition}

It can be shown (via probabilistic method) that there exists a $\eps$-PRG for width-$w$ length-$n$ ROBP with seed length $O(\log(nw/\eps))$, which is optimal. Furthermore, an \emph{explicit} PRG with such seed length would imply $\BPL=\L$. In a seminal work, Nisan~\cite{Nisan92} constructed an explicit PRG with seed length $O(\log(n)\cdot\log(nw/\eps))$, which is only a $O(\log(n))$ factor away from optimal. Nisan~\cite{Nisan94} then used this PRG to prove that any problem in $\BPL$ can be deterministically computed in $O(\log^2(n))$ space and $\poly(n)$ time. Another remarkable work by Saks and Zhou~\cite{SZ99} also applied Nisan's generator in a non-trivial way to show that any problem in $\BPL$ can be deterministically computed in $O(\log^{3/2}(n))$ space.

\subsection{Weighted PRGs}\label{subsec:WPRG}
Despite decades of effort, the seed length of Nisan's PRG remains the state-of-the-art for width $w\ge 4$. In fact, even for the $w=3$ special case, Nisan's seed length remained unbeatable until a recent work by Meka, Reingold and Tal~\cite{MRT19} which improved the seed length to $\tilde{O}(\log(n)\log(1/\eps))$. This has motivated researchers to study relaxed notions of PRGs and their applications in the derandomization of $\BPL$. A well-studied notion is that of a hitting set generator (HSG), which is the ``one-sided" variant of a PRG.
\begin{definition}[HSGs]
Let $\cF$ be a class of functions $f:\bits{n}\to\bits{}$. A $\eps$-HSG for $\cF$ is a function $G:\bits{d}\to\bits{n}$ such that for every $f\in\cF$ s.t. $\ex[x\sim\bits{n}]{f(n)}>\eps$, it holds that $\ex[s\sim\bits{d}]{f(G(s))}>0$.
\end{definition}
The study of explicit HSGs for ROBPs has a long history, starting from the seminal work by Ajtai, Koml\'{o}s and Szemer\'{e}di~\cite{AKS87}. While being weaker than PRGs, explicit constructions of HSGs can still be used to derandomize randomized log-space algorithms with one-sided error ($\RL$). In fact, a recent work by Cheng and Hoza~\cite{CH20} shows that an explicit HSG with optimal seed length $O(\log(nw/\eps))$ already implies $\BPL=\L$.   

In 2018, Braverman, Cohen and Garg~\cite{BCG20} introduced another relaxed notion of PRG called \emph{weighted PRG} (WPRG). In this relaxed notion, each output string of $G$ is further assigned a real weight that can \emph{possibly be negative}.
\begin{definition}
Let $\cF$ be a class of functions $f:\bits{n}\to\bits{}$. A $\eps$-WPRG is a pair of functions $(\rho,G):\bits{d}\to\bits{n}\times \bbR$ such that for every $f\in\cF$, 
$$\abs{\ex[x\sim \bits{n}]{f(x)}-\ex[s\sim\bits{d}]{\rho(s)\cdot f(G(s))}}\le \eps.$$
\end{definition}
Surprisingly, by simply allowing negative weights, \cite{BCG20} showed how to construct an explicit $\eps$-WPRG with seed length 
$$\tilde{O}(\log(n)\log(nw)+\log(1/\eps)),$$
which has almost optimal dependence on $\eps$. A sequence of followup work~\cite{CL20,CDRST21,PV21,Hoza21} further improved the seed length with simpler WPRG constructions. In particular, Hoza~\cite{Hoza21} completely removed the hidden $\log\log$ factors and improve the seed length to $O(\log(n)\log(nw)+\log(1/\eps))$.

It was observed in \cite{BCG20} that $\eps$-WPRGs implies $\eps$-HSGs. In addition, WPRGs seem closer to PRGs than HSGs in the sense that one can use a WPRG to estimate the expectation of a ROBP $f$ by simply enumerating all the seeds. In fact, following a suggestion in \cite{BCG20},  \cite{CL20} proved that a WPRG with good enough bound on the output of $\rho$ can be used in the derandomization framework by Saks and Zhou~\cite{SZ99}. Hoza~\cite{Hoza21} then used the WPRG in \cite{CDRST21,PV21} to prove that $\BPL$ can be derandomized in deterministic space $O(\log^{3/2}(n)/\sqrt{\log\log(n)})$. This was the first improvement over Saks and Zhou's decades-old result. 

\subsection{Regular branching programs}
For the original notion of PRGs, while there has been no improvement over Nisan's seed length for general (standard-order) ROBPs, a lot of progress has been made in some restricted families. One important example is the setting of \emph{regular ROBPs}, which is the main focus of this work.
\begin{definition}[Regular ROBPs]
We say a (standard-order) ROBP $B$ is \emph{regular} if for every transition function $B_i:[w]\times\bits{}\to {[w]}$ in $B$, every state $v\in[w]$ has exactly $2$ pre-images. 
\end{definition}
An important reason to study this family is that general ROBPs can be reduced to regular ROBPs~\cite{RTV06,BHPP22}. In fact, a surprisingly simple proof in a recent work by Lee, Pyne and Vadhan~\cite{LPV23} shows that any function that can be computed by a ROBP of length $n$ and width $w$ can also be computed by a regular ROBP of width $O(nw)$.

In 2010, Braverman, Rao, Raz and Yehudayoff~\cite{BRRY14} proved that the INW generator~\cite{INW94} with proper choices of parameters is in fact a PRG for regular ROBPs with seed length $O(\log(n)\cdot (\log\log(n)+\log(w/\eps)))$. This is better than Nisan's PRG's seed length when $\log(w/\eps)=o(\log(n))$. More generally, they introduced the ``weight" measure for ROBPs and proved that an INW generator with fixed parameters has error proportional to the weight. They then showed that regular ROBPs have smaller weight than general ROBPs when $w\ll n$, which implies their better seed length bound. (See \Cref{sec:WPRG} for the formal definitions.) Their better PRG construction for ``small-weight" ROBPs also turns out to be an important ingredient of the PRG for width-3 ROBPs in \cite{MRT19}.   

Recently, Ahmadinejad, Kelner, Murtagh, Peebles, Sidford and Vadhan~\cite{AKMPSV20} proved a remarkable result that it takes only $\tilde{O}(\log(nw))$ space to estimate the expectation of a regular ROBP $B$ in a non-black-box way. In fact, they designed an algorithm that can estimate the expectation of $B$ to a very high precision without much overhead:
\begin{theorem}\label{thm:white-box}
For every $\eps>0$ there is a deterministic algorithm which takes a regular ROBP $B$ of length $n$ and width $w$ as input, and computes a value within $\ex[x]{B(x)}\pm\eps$ in space complexity $\tilde{O}(\log(nw)\log\log(1/\eps))$. 
\end{theorem}

\subsection{Error reduction for regular branching programs}
Given the better PRG by \cite{BRRY14} in the regular setting, it is natural to ask whether one can get a better WPRG than \cite{Hoza21} in the regular setting too. This is in fact plausible because most of the WPRG constructions\footnote{This includes \cite{CDRST21,PV21,Hoza21}, and implicitly \cite{CL20} as we shall see in this paper.} introduced in \Cref{subsec:WPRG} can be viewed as a black-box \emph{error reduction procedures}: given any $\eps_0$-PRG for ROBPs for some ``mild error" $\eps_0$ (which we call the ``base PRG"), one can construct a $\eps$-WPRG for ROBPs with better dependence on $\eps$. For general standard-order ROBPs, the $O(\log(n)\log(nw)+\log(1/\eps))$ seed length described in \Cref{sec:WPRG} was obtained by taking Nisan's PRG as the base PRG. Therefore, it is natural to think that one can obtain a better $\eps$-WPRG for regular ROBPs by taking the PRG in \cite{BRRY14} as the base PRG instead.

However, it turns out that the intuition is not trivially true, because every known error reduction procedure for general ROBPs requires the ``base error" $\eps_0$ to be at most $<1/n$. When $\eps_0<1/n$, the $\tilde{O}(\log(n)\log(w/\eps_0))$ seed length bound in \cite{BRRY14} is no better than Nisan's $O(\log(n)\log(nw/\eps_0))$ seed length, so we cannot hope to get any improvement in the seed length of the corresponding WPRG.

This problem was recently solved by Chen, Hoza, Lyu, Tal and Wu \cite{CHLTW23}. They showed how to exploit the regular property and obtain a reduction from $\eps$-WPRG for regular ROBPs to PRG for regular ROBPs with error $\eps_0=O(1/\log^2(n))$. As a result, they proved the following theorem. 
\begin{theorem}[\cite{CHLTW23}]\label{thm:WPRG}
There is an explicit $\eps$-WPRG for regular ROBPs with seed length $$\tilde{O}\left(\log(n)\left(\log(w)+\sqrt{\log(1/\eps)}\right)+\log(1/\eps)\right).$$ 
\end{theorem}
Following \cite{AKMPSV20, CDRST21,PV21,Hoza21}, the WPRG construction in \cite{CHLTW23} is based on the ``inverse Laplacian" perspective of small-space derandomization and Richardson iteration. The key step in their construction is to modify the approximated inverse Laplacian based on a structure called ``shortcut graph". With the shortcut graph structure, they showed how to apply the potential argument in \cite{BRRY14} to get a better bound for $\eps$ that is still non-trivial even when $\eps_0=O(1/\log^2(n))$. Based on the same idea, \cite{CHLTW23} also showed how to get a simplified proof of the non-black-box derandomization result in \cite{AKMPSV20} (\Cref{thm:white-box}).

In short, the main purpose of using the shortcut graph idea in \cite{CHLTW23} is to embed a ``binary-recursive-like" structure into the inverse Laplacian analysis. Such a structure makes their analysis compatible with the potential argument in \cite{BRRY14}. In order to prove the non-black-box derandomization result in \Cref{thm:white-box}, \cite{CHLTW23} showed that one can apply a different potential argument based on the notion of ``singular-value approximation" (SV approximation) defined in \cite{APPSV23}. 

\subsection{Our contribution}
While the shortcut graph modification gives a nice structure to the inverse Laplacian analysis, the inverse Laplacian perspective itself is sometimes tricky to work with. In fact, although the proof of \Cref{thm:white-box} in \cite{CHLTW23} is simpler than the original proof in \cite{AKMPSV20}, they still need to work on a sophisticated matrix seminorm, and the corresponding potential argument requires  non-trivial ideas to analyze. 

In this work, we give an alternative error reduction framework for regular ROBPs by modifying a WPRG construction by Chattopadhyay and Liao~\cite{CL20}. The advantage of using \cite{CL20} is that their WPRG construction is \emph{actually binary recursive}, and hence is naturally compatible with the weight argument in \cite{BRRY14}. To construct a WPRG for regular branching program that matches the parameter in \Cref{thm:WPRG}, we show that the analysis in \cite{CL20} can be improved in the regular setting based on the weight argument in \cite{BRRY14}. Inspired by the proof of \Cref{thm:white-box} in \cite{CHLTW23}, we also give an alternative proof of \Cref{thm:white-box} based on the notion of SV approximation. Because of the binary recursive nature of \cite{CL20}, both proofs are relatively straightforward by induction and are arguably simpler than the proofs in \cite{CHLTW23}. 

In fact, our proof of \Cref{thm:white-box} implies a slightly stronger claim (\Cref{thm:SV}) which might be of independent interest: we can compute an $\eps$-SV approximation of the random walk matrix of any regular ROBP of width $w$ and length $n$ in space $\tilde{O}(\log(nw)\log\log(1/\eps))$. (See \cite{APPSV23} for comparison between SV approximation and other notions of approximation.) It is not clear how to obtain this stronger claim  from the previous proofs of \Cref{thm:white-box}~\cite{AKMPSV20, CHLTW23}. 

Finally, we show in \Cref{sec:equivalence} that the Laplacian-based construction in \cite{CHLTW23} is actually equivalent to the binary recursive construction in \cite{CL20} that we use in this paper. We note that our proofs of \Cref{thm:white-box} and \Cref{thm:WPRG} are self-contained and do not rely on this fact.  

\begin{remark}
There are two additional results in \cite{CHLTW23} which are based on their proof of \Cref{thm:white-box} and \Cref{thm:WPRG}: WPRGs for width-$3$ ROBPs and WPRGs for unbounded-width permutation ROBPs, both having seed length $\tilde{O}(\log(n)\sqrt{\log(1/\eps)}+\log(1/\eps))$. Our new proofs for \Cref{thm:white-box} and \Cref{thm:WPRG} can also be plugged into the corresponding parts of their proofs to get the same results. 
\end{remark}

\subsection{Organization}
In \Cref{sec:prelim} we introduce some general definitions that are used in both the proofs of \Cref{thm:WPRG} and \Cref{thm:white-box}, and give a brief overview of our proofs. In \Cref{sec:WPRG} we formally prove \Cref{thm:WPRG}. In \Cref{sec:white-box} we prove \Cref{thm:white-box}. 

\section{General Setup and Proof Overview}\label{sec:prelim}
\paragraph{Notation.} For $n\in\bbN$, denote $[n]=\{1,2,\ldots,n\}$. We write matrices in boldface and use $\bM[i,j]$ to denote the entry of matrix $\bM$ on the $i$-th row and the $j$-th column. We use $\bI_w$ to denote the $w\times w$ identity matrix. For a column vector $x$, we denote the $i$-th entry of $x$ by $x[i]$. For every matrix $\bM\in\bbR^{w\times w}$, $\norm[\infty]{\bM}$ denotes the infinity norm $\sup_{\norm[\infty]{v}=1} \norm[\infty]{\bM v}$ and $\norm{\bM}$ denotes the $2$-norm $\sup_{\norm{v}=1} \norm{\bM v}$. For any alphabet $\Sigma$ and string $x\in\Sigma^*$, we use $\abs{x}$ to denote the length of $x$, $x_{[i]}$ to denote the $i$-th symbol of $x$ and $x_{[\le i]}$ to denote the prefix of $x$ of length $x$. For any two strings $x,y$, we use $x\circ y$ to denote the concatenation of $x$ and $y$.

\subsection{ROBPs and matrices}
For the rest of this paper, we consider a fixed regular ROBP $B$ of length $n$ and width $w$ specified by transition functions $B_1,\ldots,B_n$. For every $i\in[n]$, and every $b\in\bits{}$, define the matrix $\bM_i(b)\in\bbR^{w\times w}$ as 
$$\forall u,v\in[w], \bM_i(b)[u,v]:=\begin{cases}1\textrm{ if }B_i(u,b)=v,\\0\textrm{ otherwise.}\end{cases}$$
We refer to $\bM_i(b)$ as the \emph{transition matrix of $B_i$ on $b$}. In addition, for every $0\le \ell < r \le n$ and a string $s\in\bits{r-\ell}$, we denote the transition matrix from layer $\ell$ to layer $r$ on input $x$ as  
$$\bM_{\ell..r}(s):=\prod_{i=1}^{r-\ell} \bM_{\ell+i}(s_i)$$
In this paper we frequently use the following fact: 
\begin{fact}
For every $\ell<m<r$ and $x\in\bits{m-\ell},y\in\bits{r-m}$, $\bM_{\ell..m}(x)\bM_{m..r}(y)=\bM_{\ell..r}(x\circ y)$.
\end{fact}
\noindent
In addition, observe that for a start state $v_0\in[w]$ and a set of accept state $V_{\mathrm{acc}}\subseteq[w]$, $B(s)=1$ if and only if there exists $v_n\in V_{\mathrm{acc}}$ s.t. $\bM_{0..n}(s)[v_0,v_n]=1$. 

Given the definitions above, we further define $\bM_i:=\frac{1}{2}(\bM_i(0)+\bM_i(1))$ which we call the \emph{random walk matrix} of $B_i$, and define $\bM_{\ell..r}:=\prod_{i=\ell+1}^r \bM_i$ which is the random walk matrix from layer $\ell$ to layer $r$. Note that $\norm[\infty]{\bM_{\ell..r}}\le 1$ because $\bM_{\ell..r}$ is right-stochastic,\footnote{This still holds even when $B$ is not regular.} and we also have $\norm{\bM_{\ell..r}}\le 1$ because $\bM_{\ell..r}$ is doubly-stochastic by the regularity. 

Finally, we define $\vst$ to be the ``start vector" s.t. $\vst[v_0]=1$ and $\vst[i]=0$ for every $i\neq v_0$, and $\ved$ to be the ``accept vector" s.t. $\ved[i]=1$ if $i\in V_{\mathrm{acc}}$ and $\ved[i]=0$ otherwise. Then observe that
$$B(s)=\vst^\top \bM_{0..n}(s)\ved$$ and $$\ex[s\in\bits{n}]{B(s)}=\vst^\top \bM_{0..n} \ved.$$
Given these facts, our goal is to find a ``good approximation" of $\bM_{0..n}$, denoted by $\widetilde{\bM_{0..n}}$, s.t. $$\abs{\vst^\top \bM_{0..n} \ved-\vst^\top \widetilde{\bM_{0..n}} \ved }\le \eps.$$ For \Cref{thm:WPRG} we want $\widetilde{\bM_{0..n}}$ to correspond to the output of a WPRG with short seed length, while for \Cref{thm:white-box} we want to make sure that $\widetilde{\bM_{0..n}}$ can be implemented in $\widetilde{O}(\log(nw)\log\log(1/\eps))$ space. Because of the different goals, the notions of approximation would also be different in the proofs of \Cref{thm:WPRG} and \Cref{thm:white-box}.

\subsection{Recursion}
In this section, we introduce a recursive definition from \cite{CL20} which we use in both the proofs of \Cref{thm:WPRG} and \Cref{thm:white-box}. Without loss of generality, we assume that $n$ is a power of $2$ for the rest of this paper. For ease of notation, we define the set of pairs $$\BS_n=\{(\ell,r):\exists i,k\in\bbN\cup\{0\}\textrm{ s.t. }\ell=i\cdot 2^k, r = \ell + 2^k \textrm{ and } 0\le \ell< r \le n\}.\footnote{$\BS$ stands for ``binary splitting".}$$
Suppose for every $(\ell_0,r_0)\in\BS_n$, we have defined a matrix $\lvl[0]{\bM_{\ell_0..r_0}}$ that is a ``mild approximation" of $\bM_{\ell_0..r_0}$. Then consider the following recursive definition of matrices for every $(\ell,r)\in\BS_n$ and every $k\in\bbN$:
\begin{equation}\label{eq:recursion}
\lvl[k]{\bM_{\ell..r}}:=
\begin{cases}
\bM_r&\textrm{ if }r-\ell=1,\\
\sum_{i+j=k} \lvl[i]{\bM_{\ell..m}} \cdot \lvl[j]{\bM_{m..r}} - \sum_{i+j=k-1} \lvl[i]{\bM_{\ell..m}} \cdot \lvl[j]{\bM_{m..r}}&\textrm{ otherwise, where $m=(\ell+r)/2$.}
\end{cases}
\end{equation}
The WPRG construction in \cite{CL20} is exactly a derandomization of the matrix $\lvl[\log(1/\eps)]{\bM_{0..n}}$, where the base cases $\lvl[0]{\bM_{\ell_0..r_0}}$ are generated by Nisan's PRG with error $1/n$. In this paper, we also prove \Cref{thm:WPRG} and \Cref{thm:white-box} by showing that $\lvl[k]{\bM_{0..n}}$ is a good enough approximation of $\bM_{0..n}$ (with different choices of the parameter $k$ and base case matrices $\lvl[0]{\bM_{\ell_0..r_0}}$).

Now for every $i\ge 0$, define $\lvl[i]{\bDel_{\ell..r}}:=\lvl[i]{\bM_{\ell..r}}-\bM_{\ell..r}$. The correctness of both \cite{CL20} and our results relies on the following identity, which was used in the proof of \cite[Lemma~15]{CL20}.
\begin{lemma}\label{lemma:identity}
For every $(\ell,r)\in\BS_n$ s.t. $r-\ell>1$ and $m=(\ell+r)/2$, 
\begin{equation*}
\lvl[k]{\bDel_{\ell..r}}=\sum_{i+j=k} \lvl[i]{\bDel_{\ell..m}} \cdot \lvl[j]{\bDel_{m..r}} - \sum_{i+j=k-1} \lvl[i]{\bDel_{\ell..m}} \cdot \lvl[j]{\bDel_{m..r}}+\lvl[k]{\bDel_{\ell..m}}\bM_{m..r}+\bM_{\ell..m}\lvl[k]{\bDel_{m..r}}.
\end{equation*}
\end{lemma}
\noindent
We briefly sketch how the correctness in \cite{CL20} was proved based on the lemma above. Suppose the ``base PRG" has error $\eps_0$ so that $\norm[\infty]{\lvl[0]{\bDel_{\ell_0..r_0}}}\le \eps_0$. Then one can prove by induction that $\norm[\infty]{\lvl[k]{\bDel_{0..n}}}\le O(n\eps_0)^{k+1}$, i.e. $\lvl[k]{\bM_{0..n}}$ is a $O(n\eps_0)^{k+1}$-approximation of $\bM_{0..n}$, using the fact that $\norm[\infty]{\bM_{\ell..r}}\le 1$ for every $\ell< r$.

Now observe that the $O(n\eps_0)^{k+1}$ bound is only non-trivial when $\eps_0<1/n$. As discussed in the introduction, the seed length of \cite{BRRY14} is not better than Nisan's PRG in this parameter regime. Therefore, in the regular setting, even if we can take the base PRG to be the improved PRG in \cite{BRRY14}, we do not get a better WPRG directly. The main contribution of this work is to give an improved analysis of the error of $\lvl[k]{\bM_{0..n}}$ in the regular setting. 

\subsection{Proof overview}\label{subsec:overview}

Similar to \cite{CHLTW23}, the reason why we can get an improvement in the regular setting is because a regular ROBP has a bounded ``total amount of mixing", no matter how large $n$ is. Our goal is to inductively prove an approximation guarantee that the error of $\lvl[k]{\bM_{\ell..r}}$ is \emph{proportional to the amount of mixing} from layer $\ell$ to layer $r$. For the proof of WPRG construction (\Cref{thm:WPRG}), this statement is formalized based on the ``weight" defined in \cite{BRRY14}. For the proof of non-black-box derandomization (\Cref{thm:white-box}), this statement is formalized with SV approximation~\cite{APPSV23}. We defer the formal definitions to later sections, and focus on why this statement gives a better bound. 

The first observation is that the last two error terms in \Cref{lemma:identity} combine nicely. That is, by induction hypothesis we can show that the second last error term $\lvl[k]{\bDel_{\ell..m}}\bM_{m..r}$ is proportional to the amount of mixing from layer $\ell$ to layer $m$, and the last error term $\bM_{\ell..m}\lvl[k]{\bDel_{m..r}}$ is proportional to the amount of mixing from layer $m$ to layer $r$. Therefore, their sum is proportional to the total amount of mixing from layer $\ell$ to layer $r$. Furthermore, we observe that with a proper choice of parameters, the error terms in the first two summations ($\sum_{i+j=k} \lvl[i]{\bDel_{\ell..m}} \cdot \lvl[j]{\bDel_{m..r}}$  and $\sum_{i+j=k-1} \lvl[i]{\bDel_{\ell..m}} \cdot \lvl[j]{\bDel_{m..r}}$) are actually very small compared to the last two terms, and hence do not affect the total error too much. 

Specifically, suppose we already know that the magnitude of $\lvl[i]{\bDel_{\ell..m}},\lvl[i]{\bDel_{m..r}}$ is roughly bounded by $\lvl[i]{\eps}$ for every $i\in\bbN$, and we want to prove by induction that the magnitude of the new error matrix $\lvl[k]{\bDel_{\ell..r}}$ is also roughly bounded by $\lvl[k]{\eps}$. We properly choose $\lvl[i]{\eps}$ as in the following lemma, so that the error terms in the first two summations sum up to roughly $\lvl[k]{\eps}/\log(n)$, which is much smaller than the ``target error" $\lvl[k]{\eps}$, and hence does not affect the total error too much. We defer the proof of \Cref{lemma:error} to \Cref{appendix:error}.\footnote{One can also choose $\lvl[i]{\eps}=\gamma^{i+1}/((2K+1)\log(n))$ where $K$ is an upper bound for $k$. Then the proof of \Cref{lemma:error} becomes straightforward, and it turns out that this does not affect the final results.}  
\begin{lemma}\label{lemma:error}
Let $\gamma<1/2$, and define $\lvl[i]{\eps}=\frac{\gamma^{i+1}}{10\log(n)(i+1)^2}$. Then for every $k\in\bbN$ we have $$\sum_{i+j=k} \lvl[i]{\eps}\lvl[j]{\eps}+\sum_{i+j=k-1} \lvl[i]{\eps}\lvl[j]{\eps} \le \lvl[k]{\eps}/\log(n).$$
\end{lemma}
With the choice of parameters above, we can prove that the error of the ``level-$k$ approximation" $\lvl[k]{\bM_{\ell..r}}$ only grows by a factor of $(1+1/\log(n))$ after each recursion. After $\log(n)$ levels of recursion, the error only grows by a constant factor. Therefore, we can choose the ``base-case error" $\lvl[0]{\eps}$ to be as small as $O(1/\log(n))$. This allows us to choose base cases with small seed length or space complexity. For the proof of \Cref{thm:WPRG}, we choose the base case to be the \cite{BRRY14} PRG with error $2^{-\sqrt{\log(1/\eps)}}$. For the proof of \Cref{thm:white-box} the base cases are generated using derandomized squaring~\cite{RV05,APPSV23}. 

\subsection{Small-space computation} 
Finally, before we start the formal proofs, we briefly discuss the model of space-bounded computation. We consider the standard model which is a Turing machine with a read-only input tape, a constant number of work tapes, and a write-only output tape. We say an algorithm runs in space $s$ if it uses at most $s$ cells on the \emph{work tapes} throughout the computation. Note that the input length and output length can be larger than $s$.

Next we recall some basic facts that we will use in space complexity analysis. For parallel composition of algorithms $\cA_1,\ldots,\cA_t$ we can reuse the work tape and get the following lemma.
\begin{lemma}\label{lemma:par-compose}
Let $\cA_1,\ldots,\cA_t$ be algorithms that on input $x$ run in space $s_1,\ldots,s_t$ respectively. Then there exists an algorithm $\cA$ that on input $x$ outputs $(\cA_1(x),\cA_2(x),\ldots,\cA_t(x))$ and runs in space $\max_{i\in[t]}(s_i)+O(\log(t))$.
\end{lemma}
Furthermore, for sequential composition $\cA_1(\cA_2(x))$, while we cannot fully store $\cA_2(x)$ in the work tape, we can still simulate an input tape containing $\cA_2(x)$ by computing the mapping $(x,i)\to \cA_2(x)_{[i]}$ instead. (See, e.g., \cite[Lemma 4.15]{AB09}.) This implies the following lemma.
\begin{lemma}\label{lemma:sq-compose}
Let $\cA_2$ be an algorithm that runs in space $s_2$ on input $x$, and $\cA_1$ be an algorithm that runs in space $s_1$ on input $\cA_2(x)$. Then there exists an algorithm $\cA$ that on input $x$ outputs $\cA_1(\cA_2(x))$ in space $s_1+s_2+O(\log(s_1+s_2+|\cA_2(x)|))$.
\end{lemma}
We also use the following lemma that can be found in \cite{MRSV17,AKMPSV20}.
\begin{lemma}\label{lemma:IMM}
Let $\bM_1,\ldots,\bM_t$ be $w\times w$ real matrices where each entry has bit length at most $T$. Then $\prod_{i=1}^t\bM_i$ can be computed in space $O(\log(t)\log(twT))$.  
\end{lemma}

\section{WPRG for regular ROBPs}\label{sec:WPRG}
Using the matrix notation, the weight defined in \cite{BRRY14} can be written as follows.\footnote{The original results in \cite{BRRY14} only consider $y\in[0,1]^w$, but one can easily generalize them to $\bbR^w$ by shifting and scaling.}
\begin{definition}
For every vector $y\in\bbR^w$ and every $i\in[n]$, define the \emph{layer-$i$ weight on $y$} as
$$W(i,y):=\sum_{u\in[w]} \sum_{b\in \bits{}} \abs{(\bM_i y)[u]-y[B_i(u,b)]}.$$
For every $0\le \ell < r \le n$, the total weight between layer $\ell$ and $r$ on $y$ is defined as 
$$W(\ell,r,y):=\sum_{i=\ell+1}^{r} W(i, \bM_{i..r} y).\footnote{For the degenerate case $i=r$, let $\bM_{r..r}$ denote the identity matrix.}$$
\end{definition}
\begin{remark}
To interpret $W(\ell,r,y)$ with the original description in \cite{BRRY14}, consider the graph view of ROBPs, and consider $y$ to be the values on the nodes in layer $r$. Then for every $i\le r$, $\bM_{i..r} y$ corresponds to the values on layer $i$. Observe that each term in the definition of $W(i,\bM_{i..r} y)$ corresponds to the ``weight" on an edge between layer $i-1$ and $i$. In consequence, $W(\ell,r,y)$ corresponds to the total weight of the \emph{sub-program between layer $\ell$ and $r$} (i.e. the ROBP specified by transition functions $(B_{\ell+1},\ldots,B_r)$). 
\end{remark}
\noindent
The following identity is straightforward by definition: 
\begin{fact}\label{lemma:weight}
For every $0\le \ell <m<r\le n$ and every vector $y\in\bbR^w$, $W(\ell,r,y)=W(\ell,m,\bM_{m..r} y)+W(m,r,y)$. This also implies $ \max(W(\ell,m,\bM_{m..r} y),W(m,r,y))\le W(\ell,r,y)$.
\end{fact}
\noindent
Given the definition of weight, the main results in \cite{BRRY14} imply the following lemmas. Note that \Cref{lemma:weight-bound} is the only place where regularity is required in this section.\footnote{To get \Cref{lemma:BRRY}, we use the fact that $\bM_{\ell..r}(\cdot)$ corresponds to the transition matrices of a ROBP of length $(r-\ell)$ and width $w$, and one can extend it to length $n$ by adding more identity transitions which do not affect the total weight.}
\begin{lemma}\label{lemma:weight-bound}
For every $\ell<r$ and every vector $y\in\bbR^w$, $W(\ell,r,y)\le w^2 \norm[\infty]{y}$.
\end{lemma}
\begin{lemma}\label{lemma:BRRY}
For every $\delta>0$, there exists an explicit PRG $G_0:\bits{d_0}\to\bits{n}$ s.t. for every $0\le \ell<r\le n$, $$\norm[\infty]{\left(\ex[s\sim\bits{d_0}]{\bM_{\ell..r}(G_0(s)_{[\le r-\ell]})}-\bM_{\ell..r}\right)y} \le \delta W(\ell,r,y).$$
In addition, the seed length is $d_0=O\left(\log(n)\left(\log\log(n)+\log(w/\delta)\right)\right)$.
\end{lemma}
Now define $W^*:=w^2$, which by \Cref{lemma:weight-bound} implies $W^*\ge \max_{y:\norm[\infty]{y}=1}\left(W(0,n,y)\right)$. To simplify notation, we define \emph{weight approximation} as follows.
\begin{definition}
For every $0\le \ell <r \le n$, we say $\widetilde{\bM_{\ell..r}}$ is a $\eps_0$-weight approximation of $\bM_{\ell..r}$ if
$$\forall y\in\bbR^w \quad \norm[\infty]{\left(\widetilde{\bM_{\ell..r}}-\bM_{\ell.r}\right)y} \le \eps_0\cdot \frac{W(\ell,r,y)}{W^*}.$$
\end{definition}
\noindent 
 Note that $\widetilde{\bM_{\ell..r}}$ being a $\delta$-weight approximation of $\bM_{\ell..r}$ also implies $\norm[\infty]{\widetilde{\bM_{\ell..r}}-\bM_{\ell..r}}\le \delta$. Now fix a parameter $\gamma>0$ to be specified later, and define $\lvl[i]{\eps}=\frac{\gamma^{i+1}}{10(i+1)^2\log(n)}$ as in \Cref{lemma:error}. Let $G_0$ be the PRG in \Cref{lemma:BRRY} with parameter $\delta=\lvl[0]{\eps}/(3W^*)$, and for every $(\ell,r)\in\BS_n$ such that $r-\ell>1$,  define $$\lvl[0]{\bM_{\ell..r}}:=\ex[s\sim\bits{d_0}]{\bM_{\ell..r}(G_0(s)_{[\le r-\ell]})},$$ which is a $(\lvl[0]{\eps}/3)$-weight approximation by \Cref{lemma:BRRY}.
Then define $\lvl[k]{\bM_{\ell..r}}$ recursively as in \Cref{eq:recursion}. The following is our main lemma for proving \Cref{thm:WPRG}:
\begin{lemma}[main]\label{lemma:WPRG-main}
For every $k\in\bbN$, every $y\in\bbR^w$ and every $(\ell,r)\in\BS_n$,
$\lvl[k]{\bM_{\ell..r}}$ is a $C_t\lvl[k]{\eps}$-weight approximation of $\bM_{\ell..r}$, where $t=\log(r-\ell)$ and $C_t=(1+1/\log(n))^t/3$.
\end{lemma}
\begin{proof}
We prove the lemma by induction over $t$ and $k$. The first base case $t=0$ is trivial since $\lvl[k]{\bM_{\ell..r}}=\bM_{\ell..r}$. The second base case $k=0$ is also true by definition. For the general case, first we note that the lemma also implies $\norm[\infty]{\lvl[k]{\bDel_{\ell.r}}}\le C_t\lvl[k]{\eps}\le \lvl[k]{\eps}$. Then observe that by \Cref{lemma:identity} and sub-additivity/sub-multiplicativity of infinity norm, we have 
\begin{align*}
\norm[\infty]{\lvl[k]{\bDel_{\ell..r}} y}
\le& \sum_{i+j\in \{k-1,k\}} \norm[\infty]{\lvl[i]{\bDel_{\ell..m}} \lvl[j]{\bDel_{m..r}}y}+\norm[\infty]{\bM_{\ell..m}\lvl[k]{\bDel_{m..r}} y} + \norm[\infty]{\lvl[k]{\bDel_{\ell..m}} \bM_{m..r}y}\\
\le& \sum_{i+j\in\{k-1,k\}} \norm[\infty]{\lvl[i]{\bDel_{\ell..m}}}\norm[\infty]{ \lvl[j]{\bDel_{m..r}}y}+\norm[\infty]{\bM_{\ell..m}}\norm[\infty]{\lvl[k]{\bDel_{m..r}} y} + \norm[\infty]{\lvl[k]{\bDel_{\ell..m}} \bM_{m..r}y}\\
\le& \sum_{i+j\in\{k-1,k\}} \left(C_{t-1}^2\lvl[i]{\eps}\lvl[j]{\eps}\cdot \frac{W(m,r,y)}{W^*}\right) +C_{t-1}\lvl[k]{\eps}\cdot\frac{W(m,r,y)+W(\ell,m,\bM_{m..r}y)}{W^*} \comment{induction}\\
\le&   \,\,C_t\lvl[k]{\eps} \cdot \frac{W(\ell,r,y)}{W^*}. 
\comment{by \Cref{lemma:error} and \Cref{lemma:weight}}.
\end{align*}
In other words, $\lvl[k]{\bM_{\ell..r}}$ is a $C_t\lvl[k]{\eps}$-weight approximation of $\bM_{\ell..r}$.
\end{proof}
\Cref{lemma:WPRG-main} shows that $\lvl[k]{\bM_{0..n}}$ is a $\lvl[k]{\eps}$-weight approximation, which also implies $\norm[\infty]{\lvl[k]{\bM_{0..n}}-\bM_{0..n}}\le \lvl[k]{\eps}$. It remains to construct a WPRG that actually ``implements" $\lvl[k]{\bM_{0..n}}$. This step is rather standard and is essentially the same as the corresponding step in \cite{CHLTW23}: to get the seed length as claimed in \Cref{thm:WPRG}, we need to further ``derandomize" $\lvl[k]{\bM_{0..n}}$ using the technique in \cite{CDRST21, PV21}. In short, we expand the recursive formula for $\lvl[k]{\bM_{0..n}}$ and get an ``error reduction polynomial" over matrices $\lvl[0]{\bM_{i..j}}$. One can show that there are at most $K=n^{O(k)}$ terms in the polynomial, and each term has at most $h=k\log(n)$ factors. Then we can use the INW generator~\cite{INW94} for length $h$ and width $w$ to approximate each term with error $\eps/2K$, which gives us a $\eps/2$-approximation to $\lvl[k]{\bM_{0..n}}$. We discuss the details in \Cref{sec:final-WPRG}. 

\begin{remark}
Note that our construction and proof also works for ``small-weight ROBPs" in general, if we define $W^*=\max_{B,y:\norm[\infty]{y}=1}\left(W(0,n,y)\right)$ instead. This only costs additional $O(\log(n)\log(W^*))$ bits of seed length. We note that this generality is important in some applications, such as the WPRG for width-$3$ ROBP in \cite{CHLTW23}. 
\end{remark}

\subsection{Final WPRG construction}\label{sec:final-WPRG}
To simplify notation, we assume without loss of generality that the first output bit of $G_0$ is unbiased, i.e. $\pr[s\in\bits{d_0}]{G_0(d_0)_{[1]}=1}=1/2$.\footnote{To get such a PRG, we can simply take a PRG $G_0'$ from \cite{BRRY14} with $(n-1)$-bit output and define $G_0(b\circ s)=b\circ G_0'(s)$, where $b$ is the first bit.} Then we can merge the two different base cases by defining $\lvl[0]\bM_{r-1..r}:=\bM_r=\lvl[k]{\bM_{r-1..r}}=\ex[s\in\bits{d_0}]{\bM_{r-1..r}(G_0(s)_{[\le 1]})}$. Now consider the following notation.
\begin{definition}\label{def:expand}
For every $0\le \ell<r\le n$, let $\mathsf{IS}_{\ell..r}$ denote the set of increasing sequences $\mathsf{sq}=(i_0,i_1\ldots,i_h)$ s.t. $\ell=i_0<i_1<\ldots<i_h=r$. We say $h$ is the \emph{length} of $\mathsf{s}$. For every $\mathsf{sq}\in\mathsf{IS}_{\ell..r}$, define 
$$\lvl{\bM}_{\mathsf{sq}}:=\prod_{j=1}^h \lvl{\bM}_{i_{j-1}..i_j}.$$

\end{definition}

\noindent
Given the notation above, we get the following lemma regarding the expansion of $\lvl[k]{\bM_{0..n}}$, which is not hard to prove by induction. For completeness we include a proof in \Cref{appendix:expansion}.
\begin{lemma}\label{lemma:expand}
For every $k\in\bbN$ and every $(\ell,r)\in\BS_n$, there is a (multi)set $S\subseteq \mathsf{IS}_{\ell..r}\times \{-1,+1\}$ which satisfies that 
\begin{itemize}
\item 
$\lvl[k]{\bM_{\ell..r}}=\sum_{(\mathsf{sq},\sigma)\in S} \sigma \lvl[0]{\bM_{\mathsf{sq}}}.$
\item 
$|S|\le (r-\ell)^{2k}$
\item 
For every $(\mathsf{sq},\sigma)\in S$, the length of $\mathsf{sq}$ is at most $k\log(r-\ell)+1$
\end{itemize}
\end{lemma}
\noindent
In addition, we would need to derandomize each term $\lvl[0]{\bM_{\mathsf{sq}}}$ using the following matrix view of INW generator~\cite{INW94}, which can be found in, e.g., \cite{BCG20}: 
\begin{lemma}\label{lemma:INW}
Let $\Sigma$ be a finite set of symbols. Suppose for every $i\in[h]$, there is a matrix-valued function $\bA_i:\Sigma \to \bbR^{w\times w}$ which on every input in $\Sigma$ outputs a stochastic matrix. Then for every $\eps_{\mathrm{INW}}>0$ there exists an explicit function $G_{\mathrm{INW}}:\bits{d}\to\Sigma^h$ such that 
$$\norm[\infty]{\ex[s\in\bits{d}]{\prod_{i=1}^h{\bA_i(G_{\mathrm{INW}}(s)_{[i]})}}-\prod_{i=1}^h\ex[x\in\Sigma]{\bA_i(x)}}\le \eps_{\mathrm{INW}},$$
and  $d=O(\log\abs{\Sigma}+\log(h)\log(hw/\eps_{\mathrm{INW}}))$.
\end{lemma}
\noindent
Now we are ready to prove \Cref{thm:WPRG}.
\begin{proof}[Proof of \Cref{thm:WPRG}]
Let $S\subseteq \mathsf{IS}_{0..n}\times [-1,1]$ be the set defined in \Cref{lemma:expand} s.t. 
$\lvl[k]{\bM_{0..n}}=\sum_{(\mathsf{sq},\sigma)\in S} \sigma\lvl[0]{\bM_{\mathsf{sq}}}$. Without loss of generality we can assume that $S$ has size exactly $2^{2k\log(n)}$ by adding dummy sequences with weight $\sigma=0$. In addition, note that there is a enumeration function $E_S:\bits{2k\log(n)}\to S$ that can be implemented in space $O(\log(k)\log(n))$ following recursive formula (\ref{eq:recursion}).\footnote{That is, we use the first $\ceil{\log(2k+1)}$ bits as index to determine a term in the recursive formula, then discard these $\ceil{\log(2k+1)}$ bits and recurse. If there's any undefined index, simply return a dummy sequence with weight $0$.} Then consider $G_{\mathrm{INW}}:\bits{d_{\mathrm{INW}}}\to \Sigma^h$ in \Cref{lemma:INW} with $\Sigma=\bits{d_0},h=k\log(n)+1$ and $\eps_{\mathrm{INW}}=\eps/(2|S|)$, then define $d=2k\log(n)+d_{\mathrm{INW}}$.
The final WPRG construction $(\rho,G):\bits{d}\to \bbR\times \bits{n} $ is as follows. On any input $s$,
\begin{enumerate}
\item 
Parse $s$ as $(s_{\mathrm{enum}}, s_{\mathrm{INW}}) \in \bits{2k\log(n)}\times \bits{d_{\mathrm{INW}}}$
\item
Define $((i_0,i_1,\ldots,i_h),\sigma):=E_S(s_{\mathrm{enum}})$.
\item 
For $j\in[h]$, define $r_j:=G_0(G_{\mathrm{INW}}(y)_{[j]})\in\bits{n}$. 
\item
Output $(\rho(s),G(s)):=(\sigma|S|, (r_1)_{[\le i_1-i_0]}\circ (r_2)_{[\le i_2-i_1]}\circ \ldots \circ (r_h)_{[\le i_h-i_{h-1}]})$.
\end{enumerate}
Next we prove the correctness of $G$. Observe that 
$$\ex[s]{\rho(s)\bM_{0..n}(G(s))}=\sum_{((i_0,\ldots,i_h),\sigma)\in S} \sigma\cdot \ex[s_{INW}]{\prod_{j=1}^h \bM_{i_{j-1}..i_j}(G_0(G_{\mathrm{INW}}(s_{\mathrm{INW}})_{[j]}))}.$$
For every term in the above equation, consider the matrix-valued functions $\bA_j:\bits{d_0}\to\bbR^{w\times w}$ s.t. $\bA_j(r)=\bM_{i_{j-1}..i_j}(G_0(r)_{[\le i_j-i_{j-1}]})$. Note that $\ex[r]{\bA_j(r)}=\lvl[0]{\bM_{i_{j-1}..i_j}}$. Then by \Cref{lemma:INW} we have 
$$\norm[\infty]{\ex[s_{INW}]{\prod_{j=1}^h \bM_{i_{j-1}..i_j}(G_{\mathrm{INW}}(s_{\mathrm{INW}})_j)} - \lvl[0]{\bM_{(i_0,i_1,\ldots,i_h)}}}\le \eps/(2|S|),$$
which by the sub-additivity of $\norm[\infty]{\cdot}$ implies 
$$\norm[\infty]{\ex[s]{\rho(s)\bM_{0..n}(G(s))}-\lvl[k]{\bM_{0..n}}}\le \eps/2.$$
We pick suitable $\gamma,k$ (to be specified later) so that  $\lvl[k]{\eps}\le \eps/2$. Then by \Cref{lemma:WPRG-main} we have $$\norm[\infty]{\ex[s]{\rho(s)\bM_{0..n}(G(s))}-\lvl[k]{\bM_{0..n}}}\le \eps.$$
Then consider the vectors $\vst,\ved\in\bbR^{w}$ corresponding to the start and end states as discussed in \Cref{sec:prelim}. Observe that $\norm[1]{\vst}=1$ and $\norm[\infty]{\ved}\le 1$ by definition. Therefore we have 
\begin{align*}
\abs{\ex[s\in\bits{d}]{\rho(s)B(G(s))}-\ex[x\in\bits{n}]{B(x)}}
&=\abs{\vst^\top\left(\ex[s]{\rho(s)\bM_{0..n}(G(s))}-\ex[x]{\bM_{0..n}(x)}\right)\ved}\\
&\le \norm[1]{\vst} \cdot \norm[\infty]{\ex[s]{\rho(s)\bM_{0..n}(G(s))}-\bM_{0..n}}\cdot \norm[\infty]{\ved}\\
&\le \eps.
\end{align*}
Finally we analyze the seed length with an unspecified parameter $0<\gamma<1/\log(n)$. Take $k$ to be the minimum integer s.t. $\lvl[k]{\eps}\le \eps/2$. Observe that $\log(1/\lvl[0]{\eps})=O(1/\gamma)$ and $k=O(\log(1/\eps)/\log(1/\gamma))$. This implies
$$d=d_0+O(\log(h)\log(hw/\eps)+k\log(n))=d_0+\tilde{O}(\log(w/\eps)+k\log(n)).$$
By \Cref{lemma:BRRY}, we have $$d_0=O(\log(n)\log(wW^*/\gamma)+\log(n)\log\log(n))=\tilde{O}(\log(n)\log(w/\gamma)).$$
Therefore,
$$d=\tilde{O}\left(\log(n)\log(w/\gamma)+\frac{\log(n)\log(1/\eps)}{\log(1/\gamma)}+\log(1/\eps)\right).$$
Taking $\gamma=2^{-\sqrt{\log(n)}}$,  we get 
$$d=\tilde{O}\left(\log(n)\left(\log(w)+\sqrt{\log(1/\eps)}\right)+\log(1/\eps)\right).$$
\noindent
Finally, observe that the space complexity is $O(d_0+\log(k)\log(n)+d_{\mathrm{INW}})=O(d)$. Therefore the WPRG is explicit.
\end{proof}

\section{Non-black-box Derandomization for Regular ROBPs}\label{sec:white-box}
We prove \Cref{thm:white-box} in this section. Inspired by \cite{CHLTW23}, we use the notion of SV approximation to capture the ``amount of mixing". To simplify notation, we define the function $D:\bbR^{w\times w}\times \bbR^w\to \bbR$ to be $D(\bA,y):=\norm{y}^2-\norm{\bA y}^2$, which plays the same role as the weight measure in the proof of \Cref{thm:WPRG}. The following fact is straightforward from the definition.
\begin{fact}
$D(\bB,y)+D(\bA,\bB y)=D(\bA\bB,y)$.
\end{fact}
\noindent
The notion of singular value approximation (SV approximation) is defined as follows.
\begin{definition}[SV approximation \cite{APPSV23}]
Let $\bW\in\bbR^{w\times w}$ be a doubly stochastic matrix. We say $\widetilde{\bW}$ is a $\eps$-SV approximation of $W$ if for every $x,y\in\bbR^{w}$,
$$\abs{x^\top (\widetilde{\bW}-\bW) y} \le \eps\cdot \left(\frac{D(\bW^\top,x)+D(\bW,y)}{2}\right).$$
Equivalently, for every $x,y\in\bbR^{w}$,
$$\abs{x^\top (\widetilde{\bW}-\bW) y} \le \eps\cdot \left(\sqrt{D(\bW^\top,x)\cdot D(\bW,y)}\right).$$
\end{definition}
\noindent 
The proof of \Cref{thm:white-box} is very similar to our proof of \Cref{thm:WPRG} in the previous section. First we also need a base case for the different approximation notion. As proved in \cite{APPSV23,CHLTW23}, there is a space-efficient implementation of SV approximation of random walk matrices based on derandomized squaring~\cite{RV05}:
\begin{lemma}[\cite{APPSV23,CHLTW23}]\label{lemma:SV-base}
For every $(\ell,r)\in\BS_n$, there is an algorithm that computes a $\delta$-SV approximation of $\bM_{\ell..r}$ in space $\tilde{O}(\log(nw)\log(1/\delta))$. Further, each entry of this approximation matrix  has bit length at most $O(\log(n)\log(1/\delta))$.
\end{lemma}
\noindent
We also need the following simple lemma, which can be found in \cite{CHLTW23}. We include its (short) proof for completeness. 
\begin{lemma}\label{lemma:SV-norm}
Suppose $\widetilde{\bW}$ is a $\delta$-SV-approximation of $\bW$, and let $\bDel=\widetilde{\bW}-\bW$. Then $$\norm[2]{\bDel y}\le \delta\sqrt{D(\bW,y)}.$$
\end{lemma}
\begin{proof}
Observe that 
$$\norm[2]{\bDel y}^2 = (\bDel y)^\top \bDel y \le \delta\sqrt{D(\bW,y)\cdot (\norm[2]{\bDel y}^2 - \norm[2]{\bW^\top \bDel y}^2)}\le\delta\sqrt{D(\bW,y)}\cdot \norm{\bDel y}.$$ 
\end{proof}
\noindent
Now for every $(\ell,r)\in\BS_n$, define $\lvl[0]{\bM_{\ell..r}}$ to be a $(\lvl[0]{\eps}/3)$-SV approximation of $\bM_{\ell..r}$, and define $\lvl[k]{\bM_{\ell..r}}$ using the recursion (\Cref{eq:recursion}). We prove the following lemma which is analogous to \Cref{lemma:WPRG-main}:
\begin{lemma}[main]
For every $k\in\bbN$ and every $(\ell,r)\in \BS_n$, 
$\lvl[k]{\bM_{\ell..r}}$ is a $C_t\lvl[k]{\eps}$-SV approximation of $\bM_{\ell..r}$, where $t=\log(r-\ell)$ and $C_t=(1+1/\log(n))^t/3$.
\end{lemma}
\begin{proof}
We again prove the lemma by induction. The base cases $t=0$ or $k=0$ are trivial by definition. For the general case, observe that  by \Cref{lemma:identity}, for every $x,y\in\bbR^w$ we have
\begin{equation}\label{eq:main}
x^\top \lvl[k]{\bDel_{\ell..r}} y 
= \sum_{i+j=k} x^\top \lvl[i]{\bDel_{\ell..m}} \lvl[j]{\bDel_{m..r}} y - \sum_{i+j=k-1} x^\top \lvl[i]{\bDel_{\ell..m}}\lvl[j]{\bDel_{m..r}} y + x^\top \bM_{\ell..m} \lvl[k]{\bDel_{m..r}}y + x^\top\lvl[k]{\bDel_{\ell..m}}\bM_{m..r} y. 
\end{equation}
To bound the first two summations in \Cref{eq:main}, observe that 
\begin{align*}
&\sum_{i+j=k} x^\top \lvl[i]{\bDel_{\ell..m}} \lvl[j]{\bDel_{m..r}} y - \sum_{i+j=k-1} x^\top \lvl[i]{\bDel_{\ell..m}}\lvl[j]{\bDel_{m..r}} y \\
&\le \sum_{i+j\in\{k-1,k\}} \norm[2]{(\lvl[i]\bDel_{\ell..m})^\top x }\norm[2]{\lvl[j]\bDel_{m..r}y } \comment{Cauchy-Schwarz}\\
&\le C_{t-1}^2\sum_{i+j\in\{k-1,k\}} \lvl[i]{\eps}\lvl[j]{\eps} \sqrt{D\left(\bM_{\ell..m}^{\top},x\right)D\left(\bM_{m..r},y\right)} \comment{\Cref{lemma:SV-norm}}\\
&\le C_{t-1}^2\cdot\frac{\lvl[k]{\eps}}{\log(n)}\cdot \frac{D\left(\bM_{\ell..m}^{\top},x\right)+D\left(\bM_{m..r},y\right)}{2} \comment{by \Cref{lemma:error} and AM-GM}\\
&\le C_{t-1}\cdot\frac{\lvl[k]{\eps}}{\log(n)}\cdot \frac{D\left(\bM_{\ell..r}^{\top},x\right)+D\left(\bM_{\ell..r},y\right)}{2} \comment{since $C_{t-1}\le 1$ and $\norm[2]{\bM_{\ell..m}},\norm[2]{\bM_{m..r}^\top} \le 1$}
\end{align*}
To bound the last two terms in \Cref{eq:main}, observe that 
\begin{align*}
&x^\top \bM_{\ell..m} \lvl[k]{\bDel_{m..r}}y + x^\top\lvl[k]{\bDel_{\ell..m}}\bM_{m..r} y\\
&\le C_{t-1}\cdot \lvl[k]{\eps}\cdot\frac{D(\bM_{m..r},y)+D(\bM_{m..r}^\top, \bM_{\ell..m}^\top x)+D(\bM_{\ell..m},\bM_{m..r}y)+D(\bM_{\ell..m}^\top, x)}{2}\\
&= C_{t-1}\cdot \lvl[k]{\eps}\cdot\frac{D(\bM_{\ell..r},y)+D(\bM_{\ell..r}^\top, x)}{2}. \comment{\Cref{lemma:weight}}
\end{align*} 
By summing up the two inequalities, we can conclude that 
$$x^\top \lvl[k]{\bDel_{\ell..r}} y \le C_t\cdot \lvl[k]{\eps}\cdot \frac{D(\bM_{\ell..r},y)+D(\bM_{\ell..r}^\top,x)}{2}.$$
Because negating $y$ does not change the bound above, we get
$$\abs{x^\top (\lvl[k]{\bM_{\ell..r}}-\bM_{\ell..r}) y} \le C_t\cdot \lvl[k]{\eps}\cdot \frac{D(\bM_{\ell..r},y)+D(\bM_{\ell..r}^\top,x)}{2},$$
i.e. $\lvl[k]{\bM_{\ell..r}}$ is a $C_t\lvl[k]{\eps}$-SV approximation of $\bM_{\ell..r}$. 
\end{proof}

Finally, let $\gamma=1/\log(n)$ and $k=O(\log(1/\eps)/\log(1/\gamma))$ be the minimum integer s.t. $\lvl[k]{\eps}\le \eps$. We claim that $\lvl[k]{\bM_{0..n}}$ can be implemented in space $\tilde{O}(\log(k)\log(nw))$, which implies the following result:
\begin{theorem}\label{thm:SV}
There is an algorithm which can compute an $\eps$-SV approximation of the random walk matrix $\bM_{0..n}$ in space $\tilde{O}(\log(nw)\log\log(1/\eps))$. 
\end{theorem}
\noindent
Note that $\Cref{thm:white-box}$ is also a direct corollary of this theorem: 
\begin{proof}[Proof of \Cref{thm:white-box}]
Compute a $(\eps/\sqrt{w})$-SV approximation of $\bM_{0..n}$, denoted by $\widetilde{\bM_{0..n}}$. By \Cref{thm:SV} this takes space $\tilde{O}(\log(nw)\log\log(1/\eps))$. Then consider $\vst,\ved$ as defined in \Cref{sec:prelim}, and output $\vst^\top \widetilde{\bM_{0..n}}\ved$. To prove the correctness, recall that $\vst^\top \bM_{0..n}\ved=\ex[x]{B(x)}$, which implies 
\begin{align*}
\abs{\vst^\top \widetilde{\bM_{0..n}}\ved-\ex[x\in\bits{n}]{B(x)}}
=\abs{\vst^\top (\widetilde{\bM_{0..n}}-\bM_{0..n})\ved}\le \eps/\sqrt{w}\norm{\vst}\norm{\ved}\le \eps
\end{align*}
by the fact that $\norm{\vst}=1$ and $\norm{\ved}\le\sqrt{w}$.
\end{proof}
\begin{remark}
Note that $\lvl[k]{\bM_{0..n}}$ does not satisfy the original definition of SV approximation in \cite{APPSV23} because it is not necessarily doubly stochastic. While every row and column in $\lvl[k]{\bM_{0..n}}$ does sum up to $1$, some of its entries might be negative. 
\end{remark}

\subsection{Space-efficient implementation}
Finally we prove that $\lvl[k]{\bM_{0..n}}$ can be implemented in space $\tilde{O}(\log(k)\log(nw))$. Note that a naive implementation of the recursion (\Cref{eq:recursion}) takes at least $O(\log(n)\log(wk))$ space. Furthermore, we cannot naively enumerate each term of $\lvl[k]{\bM_{0..n}}$ in its expansion (\Cref{lemma:expand}) either, because there are $n^{O(k)}$ terms in total, which takes at least $O(k\log(n))$ bits to enumerate. 

To reduce the space complexity, we will compute $\bM_{0..n}$ with a different recursive formula. The intuition of the new recursion is as follows. First observe that each term in the expansion of $\lvl[k]{\bM_{0..n}}$ corresponds to a way to put $k$ balls into $2n-1$ bins (indexed by $[2n-1]$), under the constraint that each odd-indexed bin contains at most one ball. To see why this is the case, observe that the original recursion corresponds to the following way to recursively enumerate all the ball-to-bin combinations: first we put $b\in\bits{}$ ball in the middle bin (which corresponds to the sign $(-1)^b$), then choose $i,j$ s.t. $i+j=k-b$, and then recursively put $i$ balls in the left $(n-1)$ bins (which corresponds to $\lvl[i]{\bM_{0..n/2}}$) and $j$ balls in the right $(n-1)$ bins (which corresponds to $\lvl[j]{\bM_{n/2..n}}$). 

Then observe that there is a different way to enumerate all the combinations with only $\ceil{\log(k)}$ levels of recursion as follows. First decide where the $h$-th ball is located, where $h=\ceil{k/2}$. If it is in an even-indexed bin, also decide how many balls are on the left and how many balls are on the right. Otherwise, there can be only one ball in the selected odd-indexed bin, and the numbers of balls on the left and right are fixed. Then for each choice, recursively enumerate the combinations on the left and right respectively. We claim that there is a corresponding recursive formula for $\lvl[k]{\bM_{0..n}}$ which can be implemented in only $\tilde{O}(\log(k)\log(nw))$ space.

To define this recursive formula, first we generalize the definition of $\lvl[k]{\bM_{\ell..r}}$ to $(\ell,r)\not\in \BS_n$. For any $(\ell,r)$ s.t. $0\le \ell< r\le n$, define $\LCA(\ell,r)$ as follows. Let $t$ be the largest integer such that there exists a multiple of $2^t$ in the range $(\ell,r)$. Then we define $\LCA(\ell,r)$ to be the unique multiple of $2^t$ in $(\ell,r)$.\footnote{If there are two consecutive multiples of $2^t$ in $(\ell,r)$, then $2^{t+1}$ divides one of them, which violates the definition of $t$.} Observe that for $(\ell,r)\in\BS_n$ s.t. $r-\ell>1$, $\LCA(\ell,r)=(\ell+r)/2$. Therefore, we can generalize the recursion (\Cref{eq:recursion}) to any $r-\ell>1$ by defining $m=\LCA(\ell,r)$. We also generalize the same recursion to $k=0,(\ell,r)\not\in\BS_n$, so that $\lvl[0]{\bM_{\ell..r}}=\lvl[0]{\bM_{\ell..m}}\lvl{\bM_{m..r}}$.\footnote{There is an empty summation term $\sum_{i+j=-1}$ which should be treated as $0$.} For the degenerate case $\ell=r$ we define $\lvl[k]{\bM_{\ell..r}}=\bI_w$. Next we want to prove the following identity which naturally gives a recursive algorithm for $\bM_{0..n}$. This identity is essentially the recursive enumeration we described above, except that we utilize $\lvl[k]{\bM_{s-1..s}}=\bM_s$ to get some cancellation which simplifies the recursion. 

\begin{lemma}\label{lemma:new-recursion}
For every $(\ell,r)$ s.t. $r>\ell$ and every $h,k\in\bbN$ s.t. $h\le k$, we have the following identity:
\begin{equation}\label{eq:new-recursion}
\lvl[k]{\bM_{\ell..r}}=\sum_{s=\ell+1}^{r}\lvl[h-1]{\bM_{\ell..s-1}}\bM_{s}\lvl[k-h]{\bM_{s..r}} - \sum_{s=\ell+1}^{r-1}\lvl[h-1]{\bM_{\ell..s}}\lvl[k-h]{\bM_{s..r}}.
\end{equation}
\end{lemma}
Surprisingly, this new recursion coincides with the recursion of Richardson iteration from the inverse Laplacian perspective. This shows that the construction in \cite{CHLTW23} is actually equivalent to the \cite{CL20} construction that we use in this paper. We briefly discuss this equivalence in \Cref{sec:equivalence}.

Before we prove the identity, we show that the new recursion does imply an algorithm that runs in space $\tilde{O}(\log(k)\log(nw))$. 
Consider the algorithm which recursively computes $\lvl[k]{\bM_{\ell..r}}$ using the formula in \Cref{lemma:new-recursion} with $h=\ceil{k/2}$. Observe that the right hand side of \Cref{eq:new-recursion} involves at most $O(n)$ matrices. In addition, given all the matrices on the right hand side, the computation of $\lvl[k]{\bM_{\ell..r}}$ takes only $O(\log(nkwT))$ additional bits, where $T$ is the maximum bit length of all the matrix entries. From \Cref{lemma:expand} and \Cref{lemma:SV-base} we can see that $T$ is at most $\tilde{O}(k\log^2(n))$, so $O(\log(nwkT))=\tilde{O}(\log(nwk))$. Finally, observe that each matrix on the right hand side has precision parameter at most $\max(h-1,k-h)\le \floor{k/2}$. Therefore the recursion reaches the $\lvl[0]{\bM_{\ell..r}}$ base cases after at most $\ceil{\log(k)}$ levels. By repeatedly applying \Cref{lemma:sq-compose} and \Cref{lemma:par-compose} we can conclude that the space complexity is at most $\tilde{O}(\log(k)\log(nw))+O(s_{\mathrm{base}})$, where $s_{\mathrm{base}}$ is the maximum space complexity of computing $\lvl{\bM_{\ell..r}}$. The base case complexity $s_\mathrm{base}$ is indeed $\tilde{O}(\log(nw))$, which we prove in \Cref{appendix:new-recursion-base}.

Finally, we prove the identity in \Cref{lemma:new-recursion}.
\begin{proof}[Proof of \Cref{lemma:new-recursion}]
We prove the claim by induction on $r-\ell$. Let $m=\LCA(\ell,r)$. For the base case $r-\ell=1$, the lemma says $\lvl[k]{\bM_{\ell..r}}=\lvl[h-1]{\bM_{\ell..\ell}}\bM_r{\bM_{r..r}}$, which is trivially true. Next we prove the general case by induction. For each matrix $\lvl[k']{\bM_{\ell'..r'}}$ on the right hand side s.t. $\ell'<m<r'$, we expand $\lvl[k']{\bM_{\ell'..r'}}$ using (\ref{eq:recursion}). Note that $\LCA(\ell',r')$ is also $m$. In addition, for the $s=m$ term in the first summation we apply the dummy expansion $\lvl[k-h]\bM_{m..r}=\sum_{a+b=k-h}\lvl[a]{\bM_{m..m}}\lvl[b]{\bM_{m..r}}-\sum_{a+b=k-h-1}\lvl[a]{\bM_{m..m}}\lvl[b]{\bM_{m..r}}$. Similarly, for the $s=m+1$ term in the first summation we also expand $\lvl[k-h]\bM_{\ell..m}=\sum_{a+b=h-1}\lvl[a]{\bM_{\ell..m}}\lvl[b]{\bM_{m..m}}-\sum_{a+b=h-2}\lvl[a]{\bM_{\ell..m}}\lvl[b]{\bM_{m..m}}$.
After rearranging we get that the right hand side equals to
\begin{align*}
&\sum_{s=\ell+1}^{m} \sum_{a+b=k-h} \lvl[h-1]{\bM_{\ell..s-1}}\bM_s\lvl[a]{\bM_{s..m}}\lvl[b]{\bM_{m..r}} 
&&-\sum_{s=\ell+1}^{m-1}\sum_{a+b=k-h} \lvl[h-1]{\bM_{\ell..s}}\lvl[a]{\bM_{s..m}}\lvl[b]{\bM_{m..r}}\\
&-\sum_{s=\ell+1}^{m} \sum_{a+b=k-h-1} \lvl[h-1]{\bM_{\ell..s-1}}\bM_s\lvl[a]{\bM_{s..m}}\lvl[b]{\bM_{m..r}}
&&+\sum_{s=\ell+1}^{m-1}\sum_{a+b=k-h-1} \lvl[h-1]{\bM_{\ell..s}}\lvl[a]{\bM_{s..m}}\lvl[b]{\bM_{m..r}}\\
&+\sum_{s=m+1}^r\sum_{a+b=h-1} \lvl[a]{\bM_{\ell..m}}\lvl[b]{\bM_{m..s-1}}\bM_s\lvl[k-h]{\bM_{s..r}}
&&-\sum_{s=m+1}^r\sum_{a+b=h-1} \lvl[a]{\bM_{\ell..m}}\lvl[b]{\bM_{m..s}}\lvl[k-h]{\bM_{s..r}}\\
&-\sum_{s=m+1}^r\sum_{a+b=h-2}\lvl[a]{\bM_{\ell..m}}\lvl[b]{\bM_{m..s-1}}\bM_s\lvl[k-h]{\bM_{s..r}}
&&+\sum_{s=m+1}^r\sum_{a+b=h-2} \lvl[a]{\bM_{\ell..m}}\lvl[b]{\bM_{m..s}}\lvl[k-h]{\bM_{s..r}}\\
& &&-\lvl[h-1]{\bM_{\ell..m}}\lvl[k-h]{\bM_{m..r}}.
\end{align*}
Note that the first summation in \Cref{eq:new-recursion} expands to the terms on the left, and the second summation in \Cref{eq:new-recursion} expands to the terms on the right.

Now we classify all the terms in the first line by $b$, and take out the right factor $\lvl[b]{\bM_{m..r}}$. For any fixed $b\le k-h $ we get the sum
$$\left(\sum_{s=\ell+1}^m\lvl[h-1]{\bM_{\ell..s-1}}\bM_s\lvl[k-b-h]{\bM_{s..m}}-\sum_{s=\ell+1}^{m-1} \lvl[h-1]{\bM_{\ell..s}}\lvl[k-b-h]{\bM_{s..m}}\right)\lvl[b]{\bM_{m..r}}.$$
Observe that this is exactly $\lvl[k-b]{\bM_{\ell..m}}\lvl[b]{\bM_{m..r}}$ by induction. Therefore, we can see that the first line is exactly $\sum_{b=0}^{k-h}\lvl[k-b]{\bM_{\ell..m}}\lvl[b]{\bM_{m..r}}$. Similarly, we can get that the second line is $-\sum_{b=0}^{k-h-1}\lvl[k-b-1]{\bM_{\ell..m}}\lvl[b]{\bM_{m..r}}$. For the third and fourth lines, we can also classify the terms by $a$ and take out the left factor $\lvl[a]{\bM_{\ell..m}}$ to get 
to get $\sum_{a=0}^{h-1}\lvl[a]{\bM_{\ell..m}}\lvl[k-a]{\bM_{m..r}}$ and $-\sum_{a=0}^{h-2}\lvl[a]{\bM_{\ell..m}}\lvl[k-a-1]{\bM_{m..r}}$ respectively. Finally, collect all the simplified terms (including the only term in the fifth line in the expansion), and we get  
\begin{align*}
&\sum_{b=0}^{k-h}\lvl[k-b]{\bM_{\ell..m}}\lvl[b]{\bM_{m..r}}+\sum_{a=0}^{h-1}\lvl[a]{\bM_{\ell..m}}\lvl[k-a]{\bM_{m..r}}-\sum_{b=0}^{k-h-1}\lvl[k-b]{\bM_{\ell..m}}\lvl[b]{\bM_{m..r}}-\sum_{a=0}^{h-2}\lvl[a]{\bM_{\ell..m}}\lvl[k-a-1]{\bM_{m..r}}-\lvl[h-1]{\bM_{\ell..m}}\lvl[k-h]{\bM_{m..r}}\\
&=\sum_{a+b=k}\lvl[a]{\bM_{\ell..m}}\lvl[b]{\bM_{m..r}} - \sum_{a+b=k-1}\lvl[a]{\bM_{\ell..m}}\lvl[b]{\bM_{m..r}}\\
&= \lvl[k]{\bM_{\ell..r}}.
\end{align*}
\end{proof}

\ifshowauthors
\section*{Acknowledgements}
We want to thank Xin Lyu, Edward Pyne, Salil Vadhan and Hongxun Wu for helpful discussions, and anonymous ITCS reviewers for many helpful comments.
\fi

\printbibliography

\appendix 
\section{Proof of  \texorpdfstring{\Cref{lemma:error}}{}}\label{appendix:error}
\begin{proof}
\begin{align*}
\sum_{i+j\in\{k-1,k\}} \lvl[i]{\eps}\lvl[j]{\eps} 
&= \frac{(k+1)^2\lvl[k]{\eps}}{10\log(n)} \cdot \left(\gamma \sum_{i+j=k}\frac{1}{(i+1)^2 (j+1)^2} + \sum_{i+j=k-1}\frac{1}{(i+1)^2(j+1)^2}\right) \\
&= \frac{(k+1)^2\lvl[k]{\eps}}{10\log(n)} \cdot \left(\gamma \sum_{i+j=k}\frac{\frac{1}{(i+1)^2}+\frac{1}{(j+1)^2}}{(i+1)^2+(j+1)^2} + \sum_{i+j=k-1}\frac{\frac{1}{(i+1)^2}+\frac{1}{(j+1)^2}}{(i+1)^2+(j+1)^2}\right)\\
&\le \frac{(k+1)^2\lvl[k]{\eps}}{10\log(n)} \cdot \left(\gamma \sum_{i+j=k}\frac{\frac{1}{(i+1)^2}+\frac{1}{(j+1)^2}}{(k+2)^2/2} + \sum_{i+j=k-1}\frac{\frac{1}{(i+1)^2}+\frac{1}{(j+1)^2}}{(k+1)^2/2}\right)\\
&\le \frac{\lvl[k]{\eps}}{10\log(n)} \cdot \left(4(1+\gamma)\sum_{i=0}^{k} \frac{1}{(i+1)^2} \right)\\
&\le \frac{\lvl[k]{\eps}}{\log(n)}.
\end{align*}
\end{proof}

\section{Proof of  \texorpdfstring{\Cref{lemma:expand}}{}}\label{appendix:expansion}

\begin{proof}
For $\mathsf{sq}_\ell=(i_0,\ldots,i_{h_\ell})\in \mathsf{IS}_{\ell..m}$ and $\mathsf{sq}_r=(j_0,\ldots,j_{h_r})\in \mathsf{IS}_{m..r}$, define $$\mathsf{sq}_\ell \| \mathsf{sq}_r=(i_0,\ldots,i_{h_\ell}=j_0,\ldots,j_{h_r}).$$
Observe that $\mathsf{sq}_\ell||\mathsf{sq}_r$ has length $h_\ell+h_r$, and that $\lvl[0]{\bM_{\mathsf{sq}_\ell||\mathsf{sq}_r}}=\lvl[0]{\bM_{\mathsf{sq}_\ell}}\cdot \lvl[0]{\bM_{\mathsf{sq}_r}}$.

With the notation above, we prove the lemma by induction over $t:=\log(r-\ell)$ and $k$. The base cases $t=0$ or $k=0$ are trivial. Consider the recursive formula (\ref{eq:recursion}), and let $\lvl[i]{S_{\ell}},\lvl[j]{S_{r}}$ denote the corresponding sets of $\lvl[i]{\bM_{\ell..m}}$ and $\lvl[j]{\bM_{m..r}}$ respectively. By distributive law, we get 
$$\lvl[k]{\bM_{\ell..r}}=\sum_{b\in\bits{}} \sum_{i+j=k-b}\sum_{(\mathsf{sq_\ell},\sigma)\in \lvl[i]{S_{\ell}}} \sum_{(\mathsf{sq_r},\sigma)\in \lvl[j]{S_{r}}} (-1)^b\sigma_\ell\sigma_r \lvl[0]{\bM_{\mathsf{sq}_\ell \| \mathsf{sq}_r}}.$$
Therefore, we can define 
$$S=\bigcup_{b\in\bits{0,1},i+j=k-b}\{(\mathsf{sq}_\ell \| \mathsf{sq}_r,(-1)^b\sigma_\ell\sigma_r):(\mathsf{sq_\ell},\sigma)\in \lvl[i]{S_{\ell}} \wedge (\mathsf{sq_r},\sigma)\in \lvl[j]{S_{r}}\},$$
which satisfies the first condition. Based on the induction hypothesis, the length of each $\mathsf{sq}_\ell \| \mathsf{sq}_r$ is at most $(i+j)(t-1)+2\le k(t-1)+2\le kt+1$, and the size of $S$ is at most 
$$\sum_{i+j\in\{k-1,k\}}\abs{\lvl[i]{S_\ell}}\abs{\lvl[j]{S_r}}\le (2k+1)2^{2k(t-1)}\le 2^{2kt}.$$

\end{proof}

\section{Base-case space complexity}\label{appendix:new-recursion-base}
In this section we prove the following claim.
\begin{lemma}
There is an algorithm which for every $\ell<r$ can compute $\lvl{\bM_{\ell..r}}$ in space $\tilde{O}(\log(nw))$. 
\end{lemma}
\begin{proof}
We claim that it is possible to output the indices $(\ell_0,r_0)\in\BS_n$ of all the factors $\lvl{\bM_{\ell_0..r_0}}$ in the expansion of $\lvl{\bM_{\ell..r}}$ in space $O(\log(n))$, and there are only $2\log(n)$ factors. The claim then follows by \Cref{lemma:SV-base}, \Cref{lemma:sq-compose} and \Cref{lemma:IMM}. 

For the base cases $(\ell,r)\in\BS_n$ the claim is straightforward. For $(\ell,r)\not\in\BS_n$, first compute $m=\LCA(\ell,r)$, and note that $\lvl{\bM_{\ell..r}}=\lvl{\bM_{\ell..m}}\lvl{\bM_{m..r}}$. For the factors of $\lvl{\bM_{\ell..m}}$, while $(\ell,m)\not\in \BS_n$, we repeatedly compute $m'=\LCA(\ell,m)$, then output $(m',m)$, and replace $m$ with the value of $m'$. Eventually when $(\ell,m)\in\BS_n$ we output $(\ell,m)$. 

Note that $\lvl{\bM_{\ell..m}}=\lvl{\bM_{\ell..m'}}\lvl{\bM_{m'..m}}$, so if we can prove that $(m',m)\in\BS_n$, then what we output is exactly the indices of the expansion (except that the output order is reversed, which is easy to deal with). To prove that $(m',m)\in\BS_n$, assume that $m'=c'\cdot 2^{t'}$ and $m=c\cdot 2^t$, for some odd integers $c',c$. In our algorithm, we always have $m'=\LCA(\ell,m)$ and $m=\LCA(\ell,r')$ for some $r'>m$. Because $\ell<m'<m<r'$, by definition of $\LCA$ we must have $t>t'$. We claim that we must have $m=m'+2^{t'}$, which implies $(m',m)\in\BS_n$. 

If this is not the case, then we must have $m'<m'+2^{t'}<m$ because $m$ is also a multiple of $m$. Then observe that $m'+2^{t'}$ is a multiple of $2^{t'+1}$, which contradicts to the fact that $m'=\LCA(\ell,m)$. 

Now we have proved how to output the expansion indices of $\lvl{\bM_{\ell.m}}$, and for $\lvl{\bM_{m..r}}$ the proof is basically the same. Finally, to prove that there are at most $2\log(n)$ factors, observe that every time we replace $m$ with $m'$, the exponent $t$ strictly decreases. Because we always have $0\le t<\log(n)$, the procedure above can repeat at most $\log(n)$ iterations. Therefore both $\lvl{\bM_{\ell.m}}$ and $\lvl{\bM_{m..r}}$ have at most $\log(n)$ factors.

\end{proof}
\section{Equivalence between the constructions in    \texorpdfstring{\cite{CHLTW23}}{} and    \texorpdfstring{\cite{CL20}}{}}\label{sec:equivalence}
The construction in \cite{CHLTW23} is based on the inverse Laplacian perspective of derandomization~\cite{AKMPSV20} and preconditioned Richardson iteration, which we briefly recap as follows. Consider the block matrix $\bW=(\bR^{w\times w})^{(n+1)\times (n+1)}$ s.t. $\bW[i-1,i]=\bM_i$ for every $i\in[n]$, and other entries of $\bW$ are $0$. The Laplacian is defined as $\bL:=\bI-\bW$, which is an invertible matrix. It can be shown that  each entry $\bL^{-1}[i,j]$ in the inverse Laplacian $\bL^{-1}$ is exactly $\bM_{i..j}$ (where $\bM_{i..i}=\bI_w$ and $\bM_{i..j}=0$ if $i>j$). 

To approximate $\bM_{0..n}$ within error $\eps$, \cite{CHLTW23} followed the preconditioned Richardson iteration approach, which first constructs a ``mild approximation" of $\bL^{-1}$ denoted by $\widetilde{\bL^{-1}}$. Their construction of $\widetilde{\bL^{-1}}$ is based on the shortcut graph structure, and one can verify that their $\widetilde{\bL^{-1}}[i,j]$ is actually the same as $\lvl{\bM_{i..j}}$ defined in this paper.

Given $\bL^{-1}$, the construction based on Richardson iteration is defined as $\lvl[k]{(\bL^{-1})}=\widetilde{\bL^{-1}}\sum_{i=0}^k (\bI-\bL \widetilde{\bL^{-1}})^i$, and $\lvl[k]{(\bL^{-1})}[0,n]$ is the final output. 
Observe that this formula satisfies the recursion $\lvl[k]{(\bL^{-1})}=\widetilde{\bL^{-1}}+\lvl[k-1]{(\bL^{-1})}(\bI-\bL\widetilde{\bL^{-1}})$. In addition, observe that for every $i<j$, 
$(\bI-\bL\widetilde{\bL^{-1}})[i,j]=\bM_{i+1}\lvl{\bM_{i+1..j}}-\lvl{\bM_{i..j}}$,\footnote{This is in fact the \emph{local consistency error} defined in \cite{CH20}, as observed in \cite{Hoza21}.} and other entries of $\bI-\bL\widetilde{\bL^{-1}}$ are $0$.
Therefore, for every $\ell\le r$ we have 
$$\lvl[k]{(\bL^{-1})}[\ell,r]=\lvl{\bM_{\ell..r}}+\sum_{s=\ell+1}^{r}\lvl[k-1]{(\bL^{-1})}[\ell,s-1]\left(\bM_s\lvl{\bM_{s..r}}-\lvl{\bM_{s-1..r}}\right).$$
Comparing it with the  special case $h=k$ of \Cref{lemma:new-recursion}:
$$
\lvl[k]{\bM_{\ell..r}}=\sum_{s=\ell+1}^{r}\lvl[k-1]{\bM_{\ell..s-1}}\bM_{s}\lvl{\bM_{s..r}} - \sum_{s=\ell+1}^{r-1}\lvl[k-1]{\bM_{\ell..s}}\lvl{\bM_{s..r}}.$$
Using the fact that $\lvl[k-1]{(\bL^{-1})}[\ell,\ell]=\bI$, it is relatively easy to prove via induction  that $\lvl[k]{(\bL^{-1})}[\ell,r]=\lvl[k]{\bM_{\ell,r}}$.

\end{document}